\documentclass[final]{l4dc2025}

\title[Meta-Learning for Adaptive Control with Automated Mirror Descent]{Meta-Learning for Adaptive Control with Automated Mirror Descent}
\usepackage{times}

\usepackage{booktabs}
\usepackage{xcolor}
\usepackage{verbatim}
\usepackage{graphicx, wrapfig, subcaption}

\newcommand{\R}{\mathbb{R}}
\newtheorem{assumption}{Assumption}

\newcommand{\argmin}{\mathop{\mathrm{argmin}}}

\coltauthor{\Name{Sunbochen Tang} \Email{tangsun@mit.edu}\\
 \Name{Haoyuan Sun} \Email{haoyuans@mit.edu}\\
 \Name{Navid Azizan} \Email{azizan@mit.edu}\\
 \addr Massachusetts Institute of Technology, Cambridge, MA}

\begin{document}

\maketitle

\begin{abstract}%
    Adaptive control achieves concurrent parameter learning and stable control under uncertainties that are linearly parameterized with known nonlinear features. Nonetheless, it is often difficult to obtain such nonlinear features. To address this difficulty, recent progress has been made in integrating meta-learning with adaptive control to learn such nonlinear features from data. However, these meta-learning-based control methods rely on classical adaptation laws using gradient descent, which is confined to the Euclidean geometry. In this paper, we propose a novel method that combines meta-learning and adaptation laws based on mirror descent, a popular generalization of gradient descent, which takes advantage of the potentially non-Euclidean geometry of the parameter space. In our approach, meta-learning not only learns the nonlinear features but also searches for a suitable mirror-descent potential function that optimizes control performance. Through numerical simulations, we demonstrate the effectiveness of the proposed method in learning efficient representations and real-time tracking control performance under uncertain dynamics.
\end{abstract}


\section{Introduction}
Autonomous systems frequently operate in complex and uncertain environments to achieve designated objectives. For example, unmanned aerial vehicles (UAVs) can be used to perform aerial imaging tasks in wildfire management, where UAVs need to navigate a pre-specified path to gather valuable data under the influence of unpredictable strong wind. In these complex environments, dynamical systems often experience unmodelled disturbances which present challenges to achieving real-time control objectives such as tracking a reference trajectory. Under the setting where the uncertain disturbance can be linearly parameterized as a linear combination of known nonlinear features and an unknown parameter vector, extensive research has been conducted in the adaptive control literature to tackle this challenge. In essence, adaptive control approaches simultaneously learn the unknown parameter and ensure asymptotic convergence to a pre-specified reference trajectory, while providing stability guarantees \citep{Narendra05,hovakimyan2010L1,slotine1991applied}. However, choosing good nonlinear features for the required linear parameterization can be difficult. 

Recently, many methods that incorporate machine learning into autonomous system control in the presence of uncertain disturbance have been proposed in the literature. In particular, meta-learning has been explored for system identification \citep{o2022neural} and control purposes \citep{richards2023control, rajeswaran2016epopt, kumar2021rma, musavi2023convergence}. 
Conceptually, meta-learning emphasizes the idea of ``learning-to-learn'' \citep{hospedales2021meta}, which aims to optimize a meta-objective across a number of different tasks. Similarly, from a task-distribution perspective, assuming that the disturbance can be linearly parameterized, controlling the system under the disturbance generated by each parameter vector can be viewed as a distinct task. The meta-objective can then be stated as discovering a control design shared by all tasks that achieves good control performance overall, despite variations in parameter values. These meta-learning-based approaches can be broadly classified into two categories: those built upon a reinforcement learning (RL) framework to meta-train a policy that maximizes a meta-objective defined as the expected reward \citep{finn2019online, rajeswaran2016epopt, kumar2021rma}, and those utilizing meta-learning to learn nonlinear features and construct an adaptive controller based on such learned features \citep{o2022neural, richards2023control}. Although RL-based frameworks can accommodate purely data-driven models, methods that incorporate meta-learning and adaptive control follow a more principled approach and provide fast online adaptive control implementations \citep{richards2023control}. 

Current approaches incorporating meta-learning into adaptive control are based on classical adaptation laws using gradient descent (GD). Proposed in \citep{lee2018natural} for robotic manipulators, a new adaptive control design leveraging adaptation laws based on mirror descent (MD) was generalized to a broad class of nonlinear systems in \citep{boffi2021implicit}. Compared to classic GD-based adaptation, the MD-based adaptation can exploit the non-Euclidean geometry of the parameter space. More specifically, by replacing the commonly used Euclidean distance in Lyapunov functions for GD-based adaptation with the more general Bregman divergence \citep{bregman1967relaxation}, MD-based adaptation ``implicitly regularizes'' the learned parameters with respect to a certain norm~\citep{azizan2019stochastic,azizan2021stochastic,sun2023unified}. Nonetheless, the effect of implicit regularization depends on the choice of the Bregman divergence function, which is non-trivial as the geometry of the unknown parameter vector space can be complicated and interconnected with the choice of nonlinear features.

In this paper, we propose a method that automates the choice of the MD-based adaptation law through meta-learning. The main contribution of the proposed framework is that we utilize meta-learning to concurrently optimize the pair of Bregman divergence choice and the nonlinear features, to achieve a meta-objective of reference tracking control. This allows for more freedom in choosing adaptation laws that yield better control performance. We have also conducted numerical experiments that demonstrate significant control performance improvement compared to our previous work in \citep{richards2023control}, which employs meta-learning to only learn the nonlinear features.


The paper is organized as follows. In Section~\ref{sec:prelim}, we provide relevant background information for meta-learning and adaptive control. 
In Section~\ref{sec:method}, we propose a method that utilizes meta-learning to automate the choice of MD-based adaptation law while achieving the meta-objective of reference tracking control. In Section~\ref{sec:theory}, we state the main theoretical results and the stability guarantees provided by our meta-learning-based adaptive controller. In Section~\ref{sec:simulation}, we present numerical simulation results that demonstrate the effectiveness of our approach on a planar quadrotor. 
In the interest of space, additional details on our method and experimental results can be found in the appendix of an extended version of our paper~\citep{tang2024meta}.

\section{Preliminaries}\label{sec:prelim}

\subsection{Meta-Learning}\label{prelim:meta_learning}
While machine learning has been widely successful in recent years, many learning algorithms require large amounts of training data and perform poorly on new tasks with limited data. To alleviate this issue, meta-learning promotes the idea of ``learning-to-learn,'' which describes the process of adapting a learning algorithm to a variety of tasks.
Meta-learning has been considered in multiple domains, such as few-shot adaptation for a new task based on priors learned from previous tasks \citep{finn2019online,almecija2022uncertainty}, hyperparameter optimization \citep{franceschi2018bilevel}, and improving generalization by searching for the algorithm that best suits a problem family \citep{thrun1998learning}.
We refer interested readers to a survey paper in \citep{hospedales2021meta} for an in-depth overview.
In this paper, we consider meta-learning in a multi-task setting and formulate it as a bi-level optimization problem, similar to \citep{hospedales2021meta}. In a standard machine learning framework, one seeks to identify the optimal parameter vector, $\theta^*$, which minimizes the real-valued loss function $l(\theta; D)$, where $D$ denotes the dataset and $\theta$ represents the model parameters. From a task-distribution view, meta-learning algorithms consist of a meta-learner which contains generic knowledge across all tasks, and a base-learner that adapts itself to each different task. More specifically, we consider a scenario where there are $M$ different tasks, corresponding to a family of loss functions $\{l_j\}_{j=1}^M$ and training datasets $\{D_j\}_{j=1}^M$. Furthermore, we partition each dataset into a training set and a validation set, i.e., $D_j = D_j^{\text{train}} \cup D_j^{\text{val}}$, where the former is used for the base-learner and the latter is used for the meta-learner. Under these notations, the training process of meta-learning can be formulated as the following bi-level optimization problem,
\begin{subequations}\label{equ:meta_formulate}
    \begin{align}
        \min_{\theta}\quad &\frac{1}{M}\left[ \sum_{j=1}^M l_j(\omega_j, D_j^{\text{val}}) + \mu_{\text{meta}} \|\theta\|^2_2 \right], \label{equ:meta_learn}\\
        \text{s.t. } \omega_j &= \mathcal{B}(\theta, D_j^{\text{train}}), \label{equ:base_learn}
    \end{align}
\end{subequations}
where $\mathcal{B}(\theta, D_j^{\text{train}})$ is the base-learner algorithm that takes the task-agnostic (meta) parameters $\theta$ and outputs task-specific parameter $\omega_j$ that adapts to the $j$th task. We choose a general algorithmic description for the base-learner in \eqref{equ:base_learn} because it is more suitable for our control-oriented downstream task, which will be explained in detail in Section~\ref{sec:method}. 
The goal of the learning problem in \eqref{equ:meta_formulate} is to find a meta parameter $\theta^*$ that, after adaption by the base-learner, provides the best average loss across all $M$ tasks.



\subsection{Adaptive Control with Mirror-Descent Adaptation}\label{sec:adap_ctrl_MD}

Consider a nonlinear control-affine system with unknown disturbance,
\begin{equation}\label{equ:affine_system}
    \dot{x} = f(x) + g(x)(u+f_{\text{ext}}(x)), 
\end{equation}
where $x \in \R^n$ is the state and $u \in \R^m$ is the control input, $f: \R^n \to \R^n$ and $g: \R^m \to \R^n$ are known functions. The function $f_{\text{ext}}: \R^n\to \R^m$ is an unknown external disturbance, known as ``matched uncertainty''  \citep{slotine1991applied}, as it can be directly canceled by the input if known a priori.

For the class of systems in \eqref{equ:affine_system}, the adaptive control literature \citep{slotine1991applied, narendra1987new} has established methods to achieve reference trajectory tracking under certain assumptions. Given a reference trajectory $x_r(t)$, in the absence of disturbance ($f_{\text{ext}} = 0$), a nominal control design $\overline{u}(t)$ can be chosen to achieve asymptotic tracking, i.e., $x(t) \to x_r(t)$ as $t\to \infty$. For the closed loop nominal system without disturbance under $\overline{u}(t)$, a control Lyapunov function for the nominal system, $\overline{V}(x)$, can be found to ensure tracking error converging to $0$ asymptotically. Under the assumption that the unknown disturbance can be linearly parameterized as $f_{\text{ext}}(x) = Y(x) a$, where $Y: \R^n \to \R^{m\times d}$ are known nonlinear feature functions and $a \in \R^d$ is an unknown vector, a standard adaptive control design which relies on both the nominal controller $\overline{u}(t)$ and the Lyapunov function $\overline{V}(x)$, can be described as follows,
\vspace{-0.25em}
\begin{subequations}\label{equ:adap_ctrl_GD}
    \begin{align}
        u &= \overline{u} - Y(x) \hat{a}, \label{equ:adap_ctrl_law_GD}\\
        \dot{\hat{a}} &= \Gamma^{-1} Y(x)^T g(x)^T \nabla_x \overline{V}(x, x_r), \label{equ:adaptation_GD}
    \end{align}
\end{subequations}
where $\Gamma \in \R^{d\times d}$ is a positive definite matrix and $\hat{a}$ is the parameter estimate. Using a Lyapunov-like function that accounts for parameter estimation error $\hat{a} - a$,
\begin{equation}\label{equ:Lyapunov_GD}
    V(x, x_r, \hat{a}, a) = \overline{V}(x, x_r) + \frac{1}{2}(\hat{a}-a)\Gamma(\hat{a}-a),
\end{equation}
the adaptive controller in \eqref{equ:adap_ctrl_GD} can be shown to ensure $\lim_{t\to\infty} x_r(t) - x(t) = 0$. The controller \eqref{equ:adap_ctrl_GD} is a ``gradient-based'' adaptive controller, as it uses $\nabla_x \overline{V}$ to construct its adaptation law. We refer the reader to Appendix~\ref{append:gd_adaptive} for more details about this controller. 

A generalization of the adaptation law \eqref{equ:adaptation_GD} was proposed for robotic manipulator systems in \citep{lee2018natural} and generalized to a broader class of nonlinear systems in \citep{boffi2021implicit}. By replacing the quadratic parameter estimation error term in \eqref{equ:Lyapunov_GD} with a more general Bregman divergence function \citep{bregman1967relaxation}, $d_{\psi}(y||x) = \psi(y) - \psi(x) - (y-x)^T\nabla \psi(x)$, a new pair of ``mirror-descent'' based adaptation law and their corresponding Lyapunov-like function are chosen as:
\begin{equation}\label{equ:adap_ctrl_law_NGD}
    \dot{\hat{a}} = -P^{-1} (\nabla^2 \psi(P\hat{a}))^{-1} P^{-1} Y(x)^T g(x)^T \nabla_x \overline{V}(x, x_r),
\end{equation}
\begin{equation}\label{equ:Lyapunov_NGD}
    V(x, x_r, \hat{a}, a) = \overline{V}(x, x_r) + d_{\psi}(Pa || P\hat{a}),
\end{equation}
where $P = P^T \succ 0$, and $\psi: \R^p \to \R$ is a strictly convex potential function.  While the Bregman divergence $d_{\psi}(\cdot||\cdot)$ is in general not a metric, it can be thought of as a generalization of the notion of ``distance'' over the Euclidean distance.
In particular, $d_{\psi}(y||x)$ measures the difference between a strictly convex potential function $\psi$ and its linear approximation around $x$. One can use direct computation to show that $V$ can be used as a stability certificate for $x(t) \to x_r(t)$ as $t\to\infty$.

\begin{remark}
    Suppose $\Gamma = P^T P$. When $\psi(x) = \frac{1}{2}x^T x$, the Bregman divergence function in \eqref{equ:Lyapunov_NGD} is reduced to $d_{\psi}(Pa||P\hat{a}) = (a-\hat{a})^T \Gamma (a-\hat{a})$. As a result, the adaptation law in \eqref{equ:adap_ctrl_law_GD} is a gradient descent-based special case of the mirror descent-based adaptation law in \eqref{equ:adap_ctrl_law_NGD}.
\end{remark}

The mirror descent-based adaptation law \eqref{equ:adap_ctrl_law_NGD} in fact exhibits an ``implicit regularization'' property that leads to a different solution than the adaptive law \eqref{equ:adaptation_GD}.
Proposition 3.1 of \citep{boffi2021implicit} showed that under the adaptation law \eqref{equ:adap_ctrl_law_NGD}, the limit of the parameter estimate $\hat{a}$ satisfies
\[ \lim_{t\to\infty} \hat{a}(t) = \argmin_{a \in \mathcal{A}} d_\psi(Pa || P\hat{a}(0)),\]
where $\mathcal{A} = \{a : f_{\text{ext}}(x(t)) = Y(x(t))a\}$ is the set of parameters that interpolates the disturbance along the current trajectory.
Therefore, the mirror descent-based adaption law updates the parameter estimates according to a geometry induced by the potential function $\psi$.

\vspace{-1em}
\section{Problem Formulation} \label{sec:formulation} 

In this paper, we consider the tracking control problems for dynamical systems under uncertain disturbances. In particular, we focus on a class of nonlinear systems that can be described by the following manipulator equation \citep{underactuated},
\begin{equation}\label{equ:man_eq}
    M(q) \Ddot{q} + C(q, \dot{q})\dot{q} + g(q) = \tau_{q, \dot{q}}(u) + f_{\text{ext}}(q, \dot{q}; w),
\end{equation}
where $M, C, g, \tau$ are known functions, $f_{\text{ext}}$ is an unknown disturbance function that depends on a disturbance parameter $w$, $q, \dot{q} \in \R^n$ are the states and $u\in \R^m$ is the control input. Additionally, we make the following assumptions about the system in \eqref{equ:man_eq}.

\vspace{-0.5em}

\begin{assumption}\label{assum:parameterization}
    The disturbance can be linearly parameterized as $f_{\text{ext}}(q, \dot{q}; w) = Y(q, \dot{q}) a$, where $Y: \R^{2n} \to \R^{n\times d}$ is an \textbf{\textit{unknown}} feature function and $a \in \R^d$ is a vector of unknown parameters.
\end{assumption}

\vspace{-1em}

\begin{assumption}\label{assum:skew-symmetric}
    $\dot{M}(q, \dot{q}) - 2C(q, \dot{q})$ is skew-symmetric.
\end{assumption}

\vspace{-1em}

\begin{assumption}\label{assum:fully-actuated}
    The system is ``fully actuated'', i.e., $\tau_{q, \dot{q}}: \mathbb{R}^{m} \to \mathbb{R}^{n}$ is surjective and invertible.
\end{assumption}

\vspace{-0.5em}

Compared to a standard adaptive control setting where the nonlinear feature function $Y(q, \dot{q})$ is known, Assumption~\ref{assum:parameterization} does not assume knowledge of $Y(q, \dot{q})$. The parameter $a$ may depend on the disturbance parameter $w$. Note also that Assumption~\ref{assum:skew-symmetric} follows from standard manipulator control literature, which is satisfied by many mechanical systems \citep{slotine1991applied}. 

Although Assumption~\ref{assum:fully-actuated} is not strictly necessary for our proposed approach to apply, it simplifies the problem setting and allows for an explicit control design with stability results. Furthermore, Assumption~\ref{assum:fully-actuated} reduces \eqref{equ:man_eq} to a special case of the control-affine system with matched uncertainty in \eqref{equ:affine_system}. Therefore, studying control designs for the manipulator equation and their stability guarantees provides a concrete example for generalizing our approach to control-affine systems in \eqref{equ:affine_system}.

In addition to the reference tracking control goals, we are particularly interested in leveraging the MD-based adaptive controller \eqref{equ:adap_ctrl_law_NGD}, which takes advantage of the geometry of the parameter space via the potential function $\psi$ in \eqref{equ:Lyapunov_NGD}. However, with the reference tracking control objective in mind, it is not clear what choice of $\psi$ and nonlinear features $Y(q, \dot{q})$ leads to the best control performance. The broader objective of this paper is to leverage the meta-learning framework to concurrently learn a suitable pair of potential function $\psi$ and nonlinear features $Y(q, \dot{q})$ from data that can achieve optimal tracking performance.


\vspace{-1em}

\section{Proposed Method: Adaptive Control with Meta-Learned Bregman Divergence}\label{sec:method}

\vspace{-0.25em}
Although recent meta-learning for control approaches in \citep{o2022neural} and \citep{richards2023control} have been proposed to address tracking control problems in this setting, these methods rely on the classical adaptation law \eqref{equ:adap_ctrl_law_GD} using gradient descent. As discussed in Section~\ref{sec:adap_ctrl_MD}, mirror-descent-based adaptation laws \eqref{equ:adap_ctrl_law_NGD} can be viewed as a strict generalization of the classical adaptation law by replacing the quadratic estimation error term in the Lyapunov function with a Bregman divergence. In this paper, we propose to leverage this generalization by automatically choosing the Bregman divergence through meta-learning. This allows us to further leverage the potentially non-Euclidean geometry of the parameter space, and it will only improve the tracking performance. 

Designing a controller $u(t)$ for the system in \eqref{equ:man_eq} under unknown disturbance $f_{\text{ext}}(q, \dot{q}; w)$ leads to a natural multi-task setting for meta-learning. Consider a set of disturbance parameters $\{w_j\}_{j=1}^M$, under Assumption~\ref{assum:parameterization}, there exists a feature function $Y(q, \dot{q})$ where we can find a set of parameter vector $\{a_j\}_{j=1}^M$ such that $f_{\text{ext}}(q, \dot{q}; w_j) = Y(q, \dot{q}) a_j$. If we view designing a tracking controller $u(t)$ for the dynamics under disturbance $f_{\text{ext}}(q, \dot{q}; w_j), j = 1, 2, ..., M$ as $M$ different tasks, then the meta-learning problem can be stated as constructing a meta-learner that learns generic knowledge across all tasks, and a base-learner that distills the task-agnostic knowledge of each task.

As discussed in Section~\ref{sec:adap_ctrl_MD}, given a feature function $Y$ and a controller $\Bar{u}$ that stabilizes the nominal dynamics, adaptive controllers with mirror-descent-based adaptation laws are able to achieve the tracking goal. If we use the MD-based adaptive controllers as base-learners, then the meta-learner needs to learn a suitable MD potential function $\psi$ and a feature function $\hat{Y}(q, \dot{q})$ that achieves good tracking performance under different disturbance $w_j$. In particular, we construct a neural network to learn unknown features $Y(q, \dot{q})$ from data, denoted as $\hat{Y}(q, \dot{q}; \theta_Y)$ where $\theta_Y$ are the weights of the network. As illustrated in Figure~\ref{fig:diagram}, we shall formalize the MD-based adaptive controller for the manipulator equation in \eqref{equ:man_eq}, and introduce the bi-level optimization formulation for our proposed meta-learning approach.

\begin{figure}[!ht]
    \centering
    \includegraphics[width=0.8\linewidth]{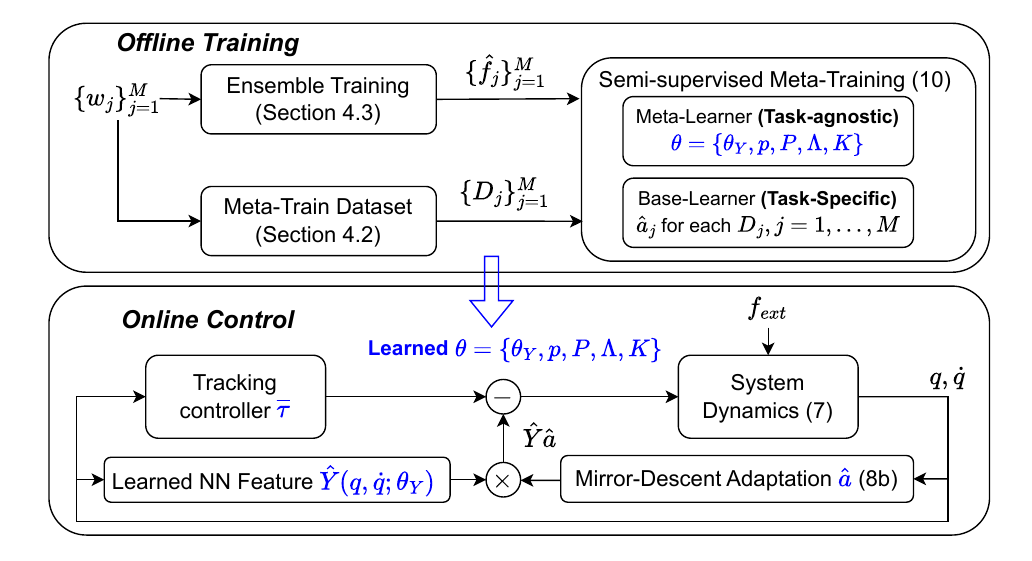}
    \caption{An illustrative diagram that shows the offline meta-learning and online adaptive control components of our proposed method. The online control components that used offline-learned parameters $\theta$ are highlighted in \textcolor{blue}{blue}.}
    \label{fig:diagram}
\end{figure}

\subsection{Base-Learner: Adaptive Control with Mirror-Descent-based Adaptation}
Given a twice-differentiable reference trajectory $q_r(t)$, following the adaptive control design for robotic manipulators introduced in \citep{slotine1987adaptive}, we define the sliding variable $s = \dot{\tilde{q}} + \Lambda \tilde{q}$ and the virtual reference trajectory $q_v:=q_r - \Lambda\int_0^t \tilde{q}(\tau) d\tau$, where $\tilde{q}(t) = q(t) - q_r(t)$ is the tracking error at time $t$, and $\Lambda \in \R^{n\times n}$ is positive definite. We propose the following MD-based adaptive controller for the system in \eqref{equ:man_eq}:
\begin{subequations}\label{equ:man_adap}
    \begin{align}
        \overline{\tau} &= M \ddot{q}_v + C \dot{q}_v + g - K s,
        \quad u = \tau_{q, \dot{q}}^{-1}(\overline{\tau} - \hat{Y}(q, \dot{q}; \theta_Y)\hat{a}), \label{equ:man_adap_ctrl}\\
        \dot{\hat{a}} &= -P^{-1} (\nabla^2 \psi(P\hat{a}))^{-1} P^{-1} \hat{Y}(q, \dot{q}; \theta_Y)^T s \label{equ:man_adap_law},
    \end{align}
\end{subequations}
where $K \in \R^{n\times n}$, $P \in \R^{d\times d}$ are both positive definite matrices. The associated Lyapunov-like function is chosen as
\begin{equation}
\label{equ:lyapunov-md}
    V = \frac{1}{2}s^T M(q)s + d_{\psi}(Pa || P\hat{a}).
\end{equation}
In particular, we choose a subclass of Bregman divergence functions defined by $\psi(a) = \|a\|_p^p = \sum_{i=1}^d |a_i|^p$, where $p\geq 1$ is a positive number, and $a_i$ is the $i$-th element of the vector $a$. Note that $\psi$ is the $\ell_p$-norm raised to the power of $p$, which provides a generalized distance metric compared to the classical adaptation laws in \eqref{equ:adaptation_GD} that is equivalent to choosing $p=2$. 

We focus on this family of Bregman divergences for two reasons. First, it provides a parameterized class of Bregman divergences, which reduces the computationally intractable task of searching for a general Bregman divergence. Second, this allows for an efficient implementation of the update rule \eqref{equ:adap_ctrl_law_NGD}, a property that has been successfully utilized in other learning problems~\citep{azizan2021stochastic,sun2023unified}. 
As previously discussed in Section~\ref{sec:adap_ctrl_MD}, adaptive controllers that use mirror-descent-based adaptations implicitly regularize the learned parameters $\hat{a}$. Thus, the optimal choice of $p$ should depend on the geometry of the parameter space. Since choosing $p$ is a non-trivial task that affects downstream control performance, we search for an optimal $p$ using the meta-learner.

Given a fixed disturbance $f_{\text{ext}}(q, \dot{q}; w_0) = Y(q, \dot{q})a_0$ with known feature function $Y(q, \dot{q})$, the controller in \eqref{equ:man_adap} is guaranteed to achieve asymptotic tracking. Formal theoretical guarantees under the meta-learning setting where $Y(q, \dot{q})$ is approximated by a neural network $\hat{Y}(q, \dot{q}; \theta_Y)$ will be discussed in Section~\ref{sec:theory}.

\subsection{Bi-level Meta-Training}\label{sec:meta-training}
Recall that the multi-task meta-learning setting considered in this paper can be formulated as a bi-level optimization problem, with its general form described in \eqref{equ:meta_formulate}. In the context of tracking control problem, each task is the control design under a disturbance model $f_{\text{ext}}(q, \dot{q}; w)$ specified by $w$. Now that we have constructed a base-learner using MD-based adaptive control, we first summarize the desired objectives for the meta-learner and introduce the bi-level meta-training problem.

For each task specified by $w_j$, the base-learner corresponds to an MD-based adaptive controller that achieves both reference trajectory tracking and identifies a suitable $\hat{a}_j$ for control purposes. The meta-parameters shared by the base-learner across different tasks include the following: the weights $\theta_Y$ in the neural network $\hat{Y}(q, \dot{q}; \theta_Y)$ for feature function approximation; the $\ell_p$-norm parameter $p$ for the MD-based adaptation; control design gain matrices $P, \Lambda, K$. Denote $\theta = \{\theta_Y, p, P, \Lambda, K\}$ as the meta-learner variables. The meta-learner tries to find the optimal $\theta$ that achieves the best average performance over all tasks.

To better explore the state space, given a disturbance parameter $w_j$, we can construct a dataset $D_j = \{w_j, \{q_r^{(ij)}(t), t \in [0, T]\}_{i=1}^N\}$ that includes the pair of $w_j$ and $N$ reference trajectories rolled out over a time window $[0, T]$ under the disturbance $f_{\text{ext}}(q, \dot{q}; w_j)$. Given $M$ different disturbance parameters $\{w_j\}_{j=1}^M$, the entire dataset can be constructed as $D = \{D_j\}_{j=1}^M$. Since the goals of the meta-learner and the base-learner are both tracking reference trajectories, we use the same dataset for training and validation, i.e., $D_j = D_j^{\text{val}} = D_j^{\text{train}}$ in \eqref{equ:meta_learn}. The loss function for each task $l_j$ is written as the sum of the average tracking error for a closed-loop system formed with the MD-based adaptive controller and squared $\ell_2$-norm of its control input. The additional regularization term prevents the use of high-gain control, which is common in control designs \citep{underactuated}. Compared to regression-oriented meta-learning methods such as \citep{o2022neural}, our proposed meta-learning approach is control-oriented, which defines the meta-objective as overall tracking performance and treats estimation on a ``need-to-know'' basis.

Now that the dataset and task loss have been constructed, we formulate the proposed meta-learning approach as the following bi-level optimization problem: 
\begin{subequations}\label{equ:meta_learning_control}
    \begin{align}
        \min_{\theta}\quad \frac{1}{M} &\left(\sum_{j=1}^M l_j \right) + \mu_{\text{meta}}\|\theta\|_2^2,\\
        \text{s.t.} \quad\quad\quad
        l_{j} &= \frac{1}{N}\left(\sum_{i=1}^N \frac{1}{T} \int_0^T (\|q^{ij}(t) - q^{ij}_r(t)\|_2^2 + \mu_{\text{ctrl}} \|u_{ij}(t)\|_2^2 dt\right),\\
        \ddot{q}^{(ij)} &= M^{-1}(\tau(u) + f_{\text{ext}}(q^{(ij)}, \dot{q}^{(ij)}; w_j) - C\dot{q}^{(ij)} - g), \label{equ:meta_learning_ODE}\\
        u^{(ij)}(t) &= \rho(q^{(ij)}(t), \dot{q}^{(ij)}(t), \theta),\\
        \dot{q}^{(ij)}(0) &= \dot{q}_r^{(ij)}(0), \quad q^{(ij)}(0) = q_r^{(ij)}(0)
    \end{align}
\end{subequations}
where $\rho(q^{(ij)}(t), \dot{q}^{(ij)}(t), \theta)$ refers to the MD-based adaptive controller given the choice of meta-learner parameters $\theta$ in \eqref{equ:man_adap}, trajectory time dependencies are abbreviated such as reducing $q^{(ij)}(t)$ to $q^{(ij)}$. In addition, the adaptation laws in \eqref{equ:man_adap_law} are initialized at $\hat{a}_j(0) = 0$ for all tasks for simplicity.

To search for the optimal parameter $\theta$, similar to most machine learning settings, we can perform stochastic gradient and back-propagation on \eqref{equ:meta_learning_control}. The task loss $l_j$ depends on a trajectory simulated using the continuous-time ODE in \eqref{equ:meta_learning_ODE}, which involves back-propagation through chosen ODE solvers such as the method presented in \citep{zhuang2020adaptive}.

\subsection{Model Ensembles}\label{sec:ensemble}

In the bi-level optimization problem, to compute each task loss $l_j$, we need forward integration of the ODE in \eqref{equ:meta_learning_ODE}, which requires access to the disturbance $f_{\text{ext}}(q, \dot{q}; w)$. During the offline meta-training process, we cannot assume explicit knowledge of the disturbance model $f_{\text{ext}}(q, \dot{q}; w)$ given a sampled disturbance parameter $w$. Inspired by approaches in the reinforcement learning literature such as \citep{clavera2018model}, we proceed to train an ensemble of neural network models that serve as surrogate disturbance models in place of $f_{\text{ext}}(q, \dot{q}; w)$ for meta-training purposes in \eqref{equ:meta_learning_control}.
Using the same the same set of sampled disturbance parameters $\{w_j\}_{j=1}^M$ used for meta-training, we sample $M$ trajectories to fit a set of surrogate noise models $\{\hat{f}_j\}_{j=1}^{M}$. The details of this step can be found in Appendix~\ref{sec:ensemble-training}.


\section{Theoretical Results: Stability Guarantees}\label{sec:theory}
In this section, we present the main theoretical guarantee. Specifically, we show that the tracking error converges asymptotically to a compact set assuming the neural network approximation error of feature $Y(q, \dot{q})$ is uniformly bounded. The key step in proving this result involves showing that the Lyapunov function~\eqref{equ:Lyapunov_NGD} is negative definite, the details of which can be found in Appendix~\ref{sec:proof-main}.

\begin{theorem}\label{thm:main}
    Assume the error between the feature function $Y$ and its neural network approximation $\hat{Y}$ is bounded, and let $\delta:= \sup_{q, \dot{q}} \|Y(q, \dot{q}) - \hat{Y}(q, \dot{q})\|$. If we use $\hat{Y}$ in place of $Y$ in the adaptive controller \eqref{equ:man_adap} and apply it to the system in \eqref{equ:man_eq}, then the tracking error $\tilde{q}(t) = q(t) - q_r(t)$ converges to a compact set $\mathcal{S} = \{\tilde{q}: \|\tilde{q}\|\leq \gamma \frac{\delta \|a\|}{\lambda_{\text{min}}(K)}\}$ where $\gamma = \int_0^{\infty} \|\exp(-\tau\Lambda) \|d\tau$. In particular, if $\delta = 0$, the tracking error $q - q_r$ converges to $0$.
\end{theorem}

\begin{remark}
    More advanced control designs like composite adaptive controller \citep{slotine1989composite} and $\sigma$-modification \citep{narendra1987new} can provide asymptotic tracking guarantees. In the literature, such robust adaptive control methods are currently only derived for adaptive controllers with classical adaptation laws. For now, we leave deriving robust MD-based adaptive control for future investigation.
\end{remark}

\section{Numerical Simulations}\label{sec:simulation}
We conduct numerical simulations of our proposed method on a fully actuated planar quadrotor model \citep{underactuated}, which follows the manipulator equation \eqref{equ:man_eq} with the specifications below:
\begin{align*}
    M(q) &= I_{3\times 3},\quad C(q, \dot{q}) = 0_{3\times 3},\quad g(q) = [0, -g, 0]^T,\\
    q &= \begin{bmatrix}
        x \\ y \\ \phi
    \end{bmatrix}, \quad 
    \tau_{q, \dot{q}}(u) = R(\phi) u = \begin{bmatrix}
        \cos\phi & -\sin\phi & 0\\
        \sin\phi & \cos\phi & 0\\
        0 & 0 & 1
    \end{bmatrix}u,
\end{align*}
where $g = -9.81$ [$m/s^2$] is the gravitational acceleration, and the system states include the planar quadrotor's center of mass position $(x, y)$ and its roll angle $\phi$ along with their time derivatives $(\dot{x}, \dot{y}, \dot{\phi})$. Note that the system is fully actuated as given any $q$, the map $\tau_{q, \dot{q}}(u) = R(\phi) u$ is invertible since $R(\phi)$ is invertible.

The unknown disturbance $f_{\text{ext}}(q, \dot{q}; w)$ is modeled as the drag force caused by a wind blowing along the inertial $x$-axis at a speed $w \in \R$, which is defined by the following equations:
\begin{align*}
    v_1 &= (\dot{x} - w)\cos\phi + \dot{y}\sin\phi,
    \quad\quad v_2 = -(\dot{x} - w) \sin\phi + \dot{y}\cos\phi,\\
    f_{\text{ext}}(q, \dot{q}; w) &= -\begin{bmatrix}
        \cos\phi & -\sin\phi \\
        \sin\phi & \cos\phi \\
        0 & 0
    \end{bmatrix} \begin{bmatrix}
        \beta_1 v_1 |v_1| \\
        \beta_2 v_2 |v_2|
    \end{bmatrix},
\end{align*}
where $\beta_1, \beta_2 > 0$ are mass-normalized drag coefficients, chosen as $0.1, 1.0$, respectively.

Given this problem setting, we then sample the trajectory sets $\{D_j\}_{j=1}^M$ and $\{\mathcal{T}_j\}_{j=1}^{M}$, which are used for meta-training and ensemble training in Sections~\ref{sec:meta-training} and \ref{sec:ensemble}, respectively. The detailed procedure for this step can be found in Appendix~\ref{sec:data-collection}.
We note that for our experiments, our training data are generated from wind speed up to 6 $[m/s]$, but we consider higher wind speed at test time.

\subsection{Benchmark with Baseline Method}
Our previous work \citep{richards2023control} proposed a meta-learning framework where the base-learner is an adaptive controller with GD-based adaptation laws, and the meta-learner includes the control gains $(\Gamma, K, \Lambda)$ and feature function neural network parameters $\theta_Y$. In comparison, our proposed method in this paper uses an adaptive controller with MD-based adaptation laws as the base-learner and leverages the meta-learner to search for an optimal $p$ which determines the potential function from the family $\psi(a) = \|a\|_p^p$.

Extensive numerical experiments in \citep{richards2023control} have already shown improvement in tracking control performance over conventional control methods like PID and regression-oriented meta-learning methods like the one proposed in \citep{o2022neural}. Considering that the GD-based adaptation law is a special case of its MD-based counterpart when $p=2$, our method is expected to show comparable or better performance over the one in \citep{richards2023control}. For these reasons, we will use the GD-based meta-adaptive controller in \citep{richards2023control} as our baseline for comparison to validate our hypothesis.

\subsection{Simulation Results}
During meta-training, we initialize the parameters $\theta_Y, P, \Lambda, K$ randomly, and initialize $p$ as $2.0$. Under the chosen Bregman divergence family $\psi(a) = \|a\|_p^p$, the MD-based adaptation law is the same as its GD-based counterpart when $p=2$, so that we can recover the baseline controller by fixing $p = 2.0$. To compare the performance of the GD-based meta-adaptive controller and our proposed method, we complete the meta-training on the same training dataset with (1) fixed $p=2.0$ (2) searching for an optimal $p$ using the meta-learner. The optimal $p$ was found to be $p=2.2$.

Using a double-loop trajectory in $x-y$ plane as a reference, we test the learned controller with $p=2.2$ and the baseline controller (with fixed $p=2.0$) under disturbances generated by wind blowing from the inertial $x$-axis at different speeds $w$ between 2.0 [$m/s$] and 10.0 [$m/s$]. We generate a trajectory of $T = 10.0$ [s] using both controllers and sample these trajectories at a fixed time-step $\Delta t = 0.02$ [s]. In Table~\ref{tab:RMS}, we present the root-mean-squared (RMS) reference trajectory tracking error $\frac{1}{N}\sum_{k=1}^N \|q(k\Delta t) - q_r(k\Delta t)\|_2^2$ for both controllers under different wind speeds. 

It is observed that our proposed method shows significant improvement over the baseline controller (roughly over 50\% reduction in tracking error) at every wind speed. Note that the training dataset contains trajectories with disturbance sampled from a distribution of $w$ supported over $[0, 6]$ [$m/s$]. It is also observed that the margin by which our proposed controller (using optimized $p=2.2$) outperforms the baseline grows as the wind speed becomes larger (i.e., further from the training distribution).

\begin{table}[!ht]
    \centering
    \begin{tabular}{cccc}
        \toprule
        $w$ & RMS: fixed $p=2.0$ & RMS: learned $p=2.2$ & In Distribution?\\
        \midrule
        \midrule
        2.0 & 0.0171 & 0.0095 & Yes\\
        \hline
        4.0 & 0.0443 & 0.0172 & Yes\\
        \hline
        6.0 & 0.0910 & 0.0303 & No\\
        \hline
        8.0 & 0.1738 & 0.0483 & No\\
        \hline
        10.0 & 0.2746 & 0.0714 & No\\
        \toprule
    \end{tabular}
    \caption{Tracking RMS under different wind speeds. Compared to the GD-based method~\citep{richards2023control}, our proposed MD-based approach not only achieves lower tracking error at every wind speed but also generalizes significantly better for wind speeds outside of our training data.}
    \label{tab:RMS}
\end{table}

\begin{figure}[!ht]
    \centering
    \newlength{\imagewidth}
    \settowidth{\imagewidth}{\includegraphics{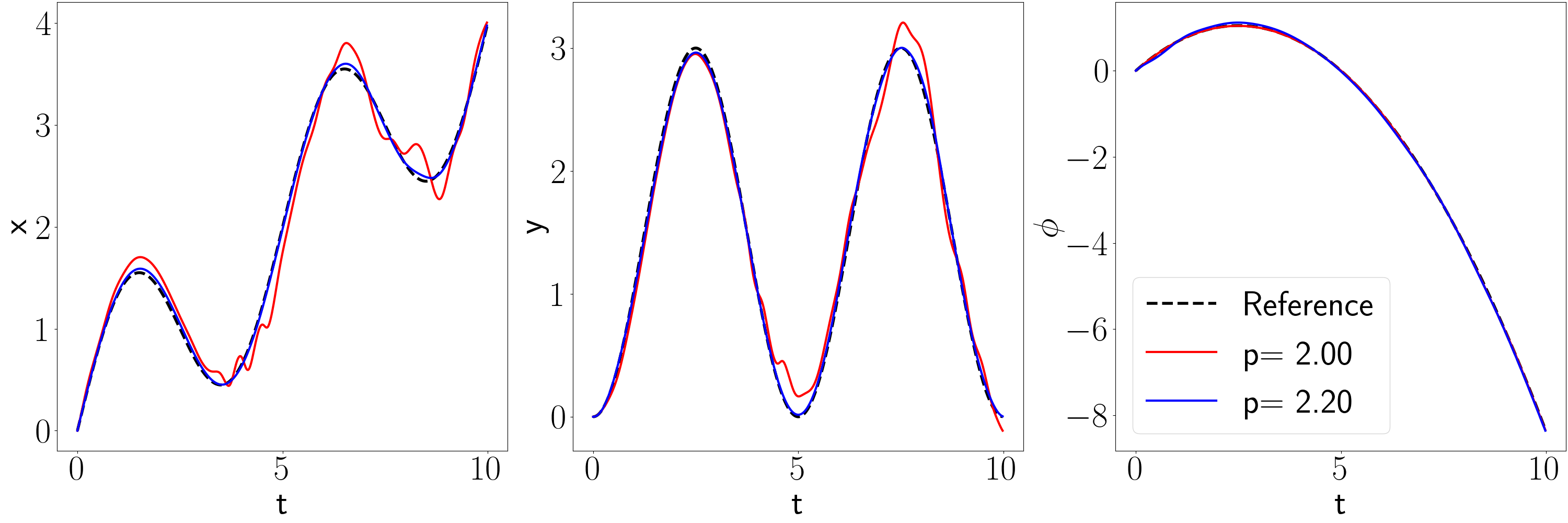}}
    \subfigure[Roll angle history]{\includegraphics[width=0.25\textwidth, trim=0.66\imagewidth{} 0 0 0, clip]{plots/w_8.0_q_i.png}}
    \subfigure[$x-y$ Phase plot]{\includegraphics[width=0.74\textwidth]{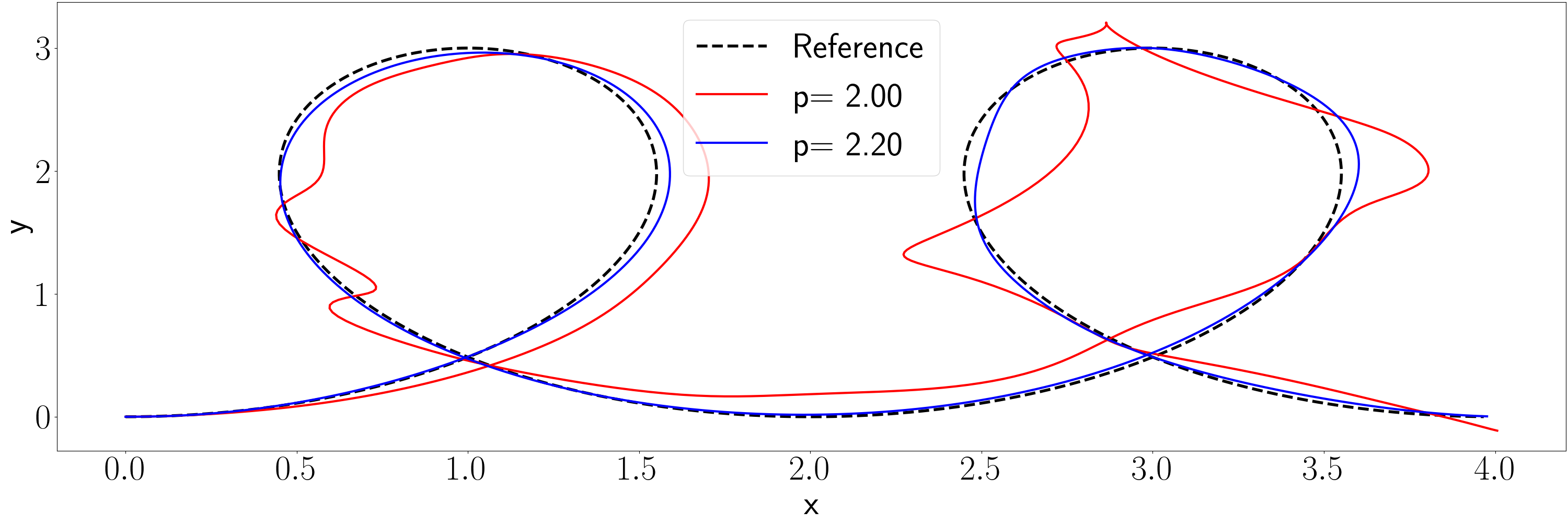}}
    \caption{Simulation result under disturbance $w=8.0$ [$m/s$]. These figures demonstrate that our MD-based controller (\textcolor{blue}{blue}) improves the tracking accuracy significantly over the GD-based baseline controller.}
    \label{fig:w_8}
\end{figure}

In Figure~\ref{fig:w_8}, we provide a comparison of our method against the baseline with an out-of-distribution disturbance $w = 8.0$ [$m/s$].
Figure~\ref{fig:w_8}(a) illustrates the time history of the roll angle, and Figure~\ref{fig:w_8}(b) shows the $x-y$ phase plot. As we can see, the learned MD-based controller with $p=2.2$ tracks the reference well while the baseline controller fails to track the reference after the first loop.
We also include plots for additional wind speeds in Appendix~\ref{sec:additional-fig}.

\section{Conclusions}\label{sec:conclusion}
In this paper, we proposed a method that leverages meta-learning to automate the choice of the MD-based adaptive control. Compared to other methods that combine meta-learning and adaptive control, our method utilizes meta-learning to search for a potential function used in the MD-based adaptation law while learning a neural network representation of unknown disturbance, allowing us to take advantage of the (potentially non-Euclidean) geometry of the parameter space. Numerical simulations conducted on a planar quadrotor demonstrate improved control performance of our method compared to previous 
methods based on classical adaptation laws.

\clearpage

\acks{
The authors acknowledge the MIT SuperCloud and Lincoln Laboratory Supercomputing Center for providing computing resources that have contributed to the results reported within this paper. This work was supported in part by MathWorks, the MIT-IBM Watson AI Lab, the MIT-Amazon Science Hub, and the MIT-Google Program for Computing Innovation.
}

\bibliography{ref}

\clearpage

\appendix


\section{Background on Adaptive Control with Gradient-based Adaptation}\label{append:gd_adaptive}
This section complements Section~\ref{sec:adap_ctrl_MD} and covers the topic of adaptive control with gradient descent-based adaptation in greater detail.
Recall that we consider a control-affine system according to \eqref{equ:affine_system}:
\begin{equation*}
    \dot{x} = f(x) + g(x)(u+f_{\text{ext}}(x)).
\end{equation*}

Our objective is to track a desired reference trajectory $x_r(t)$. Even without any disturbance $f_{\text{ext}} = 0$, tracking an arbitrary reference is not possible unless the system is fully actuated, in other words, $g(x)$ is invertible \citep{underactuated}. Therefore, we consider tracking a reference trajectory that is open-loop feasible, meaning that there exists $(x_r, u_r)$ satisfying the following nominal dynamics in the absence of disturbance,
\vspace{-0.5em}
\begin{equation}\label{equ:affine_system_nodis}
    \dot{x} = f(x) + g(x)u.
\end{equation}
Assuming the disturbance $f_{\text{ext}}$ is known, any controller $\overline{u}$ that achieves tracking for the nominal dynamics in \eqref{equ:affine_system_nodis} can lead to a controller $u = \overline{u}+f_{\text{ext}}(x)$ that ensures $x_r(t) - x(t) \to 0$ by directly canceling the disturbance $f_{\text{ext}}$.

Given the complex nature of nonlinear dynamics, it is difficult to give a closed-form expression for $\overline{u}$, a tracking controller for the nominal dynamics, since their designs are typically based on the specific dynamical system structure and computational resource constraints. Approaches to designing $\overline{u}$ include a variety of nonlinear control methods, e.g., control-contraction-metrics-based feedback \citep{manchester2017control}, sliding mode control \citep{slotine1991applied}, and model predictive control \citep{kouvaritakis2016model}. Despite the variety of control laws developed, they are all associated with a ``Lyapunov-like'' scalar function $\overline{V}(x, x_r)$ that guarantees asymptotic stability, or in this case tracking objective $x_r - x \to 0$. Furthermore, different control designs may use different notions of stability, e.g., Lyapunov's direct method \citep{lyapunov1992general}, LaSalle's invariance principle \citep{lasalle1960some}, or incremental stability \citep{lohmiller1998contraction}. Therefore, it is difficult to establish a uniform stability condition on $\overline{V}$ that establishes stability guarantees. In Section~\ref{sec:method}, we give a concrete example of fully actuated dynamical systems.

When the disturbance $f_{\text{ext}}$ is unknown, we make an additional assumption that the disturbance $f_{\text{ext}}(x)$ can be linearly parameterized. Specifically, there are known nonlinear features $Y: \R^n \to \R^{m\times d}$ and an unknown but fixed vector $a \in \R^d$ such that $f_{\text{ext}}(x) = Y(x) a$. 
Recall that in \eqref{equ:adap_ctrl_GD}, we define  the conventional adaptive control design
    \begin{align*}
        u &= \overline{u} - Y(x) \hat{a}, \\
        \dot{\hat{a}} &= \Gamma^{-1} Y(x)^T g(x)^T \nabla_x \overline{V}(x, x_r).
    \end{align*}
Given a tracking controller for the nominal dynamics $\overline{u}$ and its stability certificate $\overline{V}(x, x_r)$, under the linear parameterization and matched uncertainty settings, this controller can be shown to ensure asymptotic tracking  under disturbance in the following lemma.

\begin{lemma}
\label{thm:lyapunov-gd}
Under the adaptive controller defined in \eqref{equ:adap_ctrl_GD}, the Lyapunov-like function
\begin{equation*}
    V(x, x_r, \hat{a}, a) = \overline{V}(x, x_r) + \frac{1}{2}(\hat{a}-a)\Gamma(\hat{a}-a)
\end{equation*}
serves as a stability certificate that ensures $x(t) \to x_r(t)$ as $t\to \infty$.
\end{lemma}

\begin{proof}
Note that, since $\Gamma$ is a positive matrix, the Lyapunov function
\begin{equation*}
    V(x, x_r, \hat{a}, a) = \overline{V}(x, x_r) + \frac{1}{2}(\hat{a}-a)\Gamma(\hat{a}-a),
\end{equation*}
is positive definite whenever $\bar{V}$ is positive definite.
Applying the adaptive controller in \eqref{equ:adap_ctrl_GD} to the dynamics in \eqref{equ:affine_system}, we get
\begin{align*}
    \dot{V} \! = \! \nabla_x \overline{V}(x, x_r)^T \! (f(x) + g(x) \overline{u}) \! + \! \nabla_{x_r} \overline{V}(x, x_r)^T \! (f(x_r) \! + \! g(x_r)u_r)
\end{align*}
Note that the above time derivative is the same as $\dot{\overline{V}}$ when $\overline{u}$ is applied to the nominal dynamics $\dot{x} = f(x) + g(x)u$. Using Barbalat's lemma, $V(x, x_r, \hat{a}, a)$ can serve as a certificate that ensures $x \to x_r$.

\end{proof}

\section{Additional Details on Ensemble Training}
\label{sec:ensemble-training}
As previously discussed in Section~\ref{sec:ensemble}, we cannot assume direct access to the disturbance model $f_{\text{ext}}(q, \dot{q}; w)$ in our offline meta-training and instead we have to utilize an ensemble of surrogate models in place of $f_{\text{ext}}(q, \dot{q}; w)$.
To this end, in this section, we discuss how we may train such surrogate models.

Recall that model ensemble training and meta-training share the same set of sampled disturbance parameters $\{w_j\}_{j=1}^M$. With a given controller, we collect training data over $M$ trajectories of length $T_e$, each generated by the closed-loop system under disturbance $f_{\text{ext}}(q, \dot{q}; w_j)$ for $j = 1, \dots, N$. Note also that the controller for trajectory generation does not need to be performant or achieve any specific goal, for example, a general PID controller would be sufficient.  From each simulated trajectory, we can acquire the state and control input history $(q_j(t), \dot{q}_j(t)), u_j(t), t \in [0, T_e]$, respectively. By sampling from these $M$ trajectories at a fixed time step, $\Delta t = T/N_e$, the trajectory data from a single trajectory under disturbance $w_j$ is collected as $\mathcal{T}_j = \{q_j^{(k)}, \dot{q}_j^{(k)}, u_j^{(k)}\}_{k=1}^{N_e}$, where $(\cdot)_j^{(k)} = (\cdot)(k\Delta t)$ represents a quantity sampled at the $k$-th time step.

By constructing a neural network $\hat{f}(q, \dot{q}; \xi)$ as a surrogate model for the unknown disturbance model $f_{\text{ext}}(q, \dot{q}; w_j)$, we fit a neural network model $\hat{f}(q, \dot{q}; \xi_j)$ for each trajectory data $\mathcal{T}_j$. Hence we obtain an ensemble of models $\hat{f}_j = \hat{f}(q, \dot{q}; \xi_j)$ that roughly approximate the distribution of $f(q, \dot{q}; w)$ over $w$. Each model $\hat{f}(q, \dot{q}; \xi_j)$ is fit separately to its corresponding dataset $\mathcal{T}_j$ by performing stochastic gradient descent on the following optimization problem:
\begin{align*}
    &\min_{\xi_j} \frac{1}{N_e} \left( \sum_{k=1}^{N_e-1} \left\|\begin{bmatrix}
        \hat{q}_j^{(k)}(\Delta t) - q_j^{(k+1)}\\
        \dot{\hat{q}}_j^{(k)}(\Delta t) - \dot{q}_j^{(k+1)}
    \end{bmatrix}\right\|_2^2  \right)\\
    \text{s.t. } \quad
        &\ddot{\hat{q}}_j^{(k)}(t) = M^{-1}[\tau(u_j^{(k)}(t)) + \hat{f}(\hat{q}_j^{(k)}(t), \dot{\hat{q}}_j^{(k)}(t); \xi_j) - C\dot{\hat{q}}^{(k)}_j(t), \quad t \in [0, \Delta t],\\
        &\dot{\hat{q}}_j^{(k)}(0) = \dot{q}_j^{(k)}, \quad \hat{q}_j^{(k)}(0) = q_j^{(k)}, \quad k = 1,2, ..., N_e-1.
\end{align*}
In the optimization problem above, the model $\hat{f}(q, \dot{q}, \xi_j)$ is fitted to $\mathcal{T}_j$ through a one-step prediction setting: $(\hat{q}_j^{(k)}(t), \dot{\hat{q}}_j^{(k)}(t), \ddot{\hat{q}}_j^{(k)}(t)), t \in [0, \Delta t]$ is the numerically integrated trajectory initialized from true trajectory samples at $k\Delta t$, $(q_j^{(k)}, \dot{q}_j^{(k)})$.

\section{Training Data Collection}
\label{sec:data-collection}

As discussed in Section~\ref{sec:method}, training data needs to be collected for both ensemble training and meta-training. For ensemble training, we generate the trajectory set $\{\mathcal{T}_j\}_{j=1}^{M}$ by the following steps:
\begin{enumerate}
    \item Sample $w_j$ from a distribution $w_j \sim W_{\text{train}}$.
    \item Generate a random walk reference waypoints $q_r(k\Delta t)$, and fit a smooth spline between each waypoint, resulting in a reference trajectory $q_r(t), t\in [0, T]$.
    \item Apply a PID controller to the system for the purpose of tracking reference $q_r(t)$, obtain the simulated trajectory from the closed-loop system $q(t), \dot{q}(t), t\in [0, T]$.
    \item Acquire $\mathcal{T}_j$ by sampling from the simulated trajectory  as described in Section~\ref{sec:ensemble}.
\end{enumerate}

Recall from Section~\ref{sec:meta-training} that the meta-training dataset $\{D_j\}_{j=1}^M$ consists of $M$ different pairs of disturbance parameter $w_j$ and reference trajectories, i.e., $D_j = \{w_j, \{q_r^{(ij)}(t), t \in [0, T]\}_{i=1}^N\}$. Similar to the procedure introduced above, we follow steps 1 and 2 to acquire each pair $D_j$ for meta-training, except that in step 2, we need to generate $N$ different reference trajectories to obtain $\{q_r^{(ij)}\}_{i=1}^N$ for every $w_j$.

The training distribution $W_{train}$ for both ensemble training and meta-training is chosen as a scaled beta distribution. More precisely, a random variable following $W_{train}$, denoted as $w$, is defined as $w = w_{max} \xi$, where $\xi\sim Beta(\alpha, \beta)$ follows a Beta distribution. We choose $w_{max} = 6$, $\alpha = 5, \beta = 9$, resulting in a distribution finitely supported over $[0, 6]$, which is consistent with our simulation setup where the wind speed $w_j$ is often bounded.

\section{Proof of Theorem~\ref{thm:main}}
\label{sec:proof-main}

    By definition of the Bregman divergence,
    \begin{align*}
        \frac{d}{dt} d_{\psi} (P a || P\hat{a}) &= \frac{d}{dt}\left( \psi(Pa) - \psi(P\hat{a}) -\nabla\psi(P\hat{a})^T(Pa - P\hat{a}) \right)\\
        &= -\nabla \psi(P\hat{a})^T P \dot{\hat{a}} + (\nabla^2 \psi(P\hat{a})P\dot{\hat{a}})^T P \tilde{a} + \nabla \psi(P\hat{a})^T P\dot{\hat{a}}\\
        &= (\nabla^2 \psi(P\hat{a})P\dot{\hat{a}})^T P \tilde{a}
    \end{align*}
    
    By applying the adaptive controller in \eqref{equ:man_adap_ctrl} to the system in \eqref{equ:man_eq}, recall that $s = \dot{\tilde{q}} + \Lambda \tilde{q}$, $q_v = q_r - \Lambda \int_0^t \tilde{q}(\tau) d\tau$, the resulting closed-loop dynamics can be written as
    \begin{align*}
        M\ddot{q} + C\dot{q} + g - f_{\text{ext}}(q, \dot{q}) &= M \ddot{q}_v + C \dot{q}_v + g - K s - \hat{Y}(q, \dot{q})\hat{a}\\
        \implies M \dot{s} + (C+K)s &= Y(q, \dot{q})a - \hat{Y}(q, \dot{q})\hat{a}
    \end{align*}
    where $\tilde{a} = \hat{a} - a$ is the parameter estimation error. Using the skew-symmetric $\dot{M} - 2C$ assumption, we have
    \begin{align*}
        \dot{V} &= \frac{1}{2}s^T \dot{M} s + s^T (-(C+K)s + Ya - \hat{Y}\hat{a}) + \frac{d}{dt} d_{\psi} (P a || P\hat{a})\\
        &= -s^T K s + s^T Y \tilde{a} + s^T(\hat{Y} - Y)a + \frac{d}{dt} d_{\psi} (P a || P\hat{a}).
    \end{align*}
    Using the adaptation law in \eqref{equ:man_adap_law},
    \begin{align*}
        \dot{V} = -s^T K s - s^T Y\tilde{a} + s^T(\hat{Y} - Y)a + s^T Y P^{-1}P\tilde{a}= -s^T K s + s^T (\hat Y - Y)a
    \end{align*}
    Let $\lambda_{\text{min}}(K)$ be the minimum eigenvalue of $K$, define the set $E = \{s: \|s\|\leq \frac{\delta \|a\|}{\lambda_{\text{min}}(K)}$\}, then
    \begin{align*}
        \dot{V} &\leq -\lambda_{\text{min}}(K) \|s\|^2 + \|s\|\delta \|a\| < 0, \forall s \notin E.
    \end{align*}
    Therefore, $s(t)$ will enter the set $E$ in finite time and stay there afterwards. As a result, $s:\R \to \R^n$ as a function of time satisfies $s \in \mathcal{L}_{\infty}$ and $\|s\|_{\infty} = \sup_{t} \|s(t)\| \leq \delta \|a\|/\lambda_{\text{min}}(K)$.

    Since $\Lambda \succ 0$, $\dot{\tilde{q}}(t) = -\Lambda \tilde{q}(t) + s(t)$ can be seen as an internally stable linear system with an input $s(t)$. Using the results from input-to-state stability analysis \citep{Sontag2021}, we have:
    \begin{align*}
        \|\tilde{q}(t)\| \leq \|\exp(-\Lambda t)\|\|\tilde{q}(0)\| + \gamma \|s\|_{\infty}, 
        \gamma = \int_0^{\infty} \|\exp(-\tau\Lambda) \|d\tau < \infty.
    \end{align*}
    Therefore, $\tilde{q}$ converges to the ball $B(0, \frac{\gamma \delta \|a\|}{\lambda_{\text{min}}(K)})$ in finite time and stays in it afterwards.

    When $\delta = 0$, $E = \{0\}$ hence $s(t) \to 0$ as $t\to \infty$. Since $\dot{\tilde{q}} = -\Lambda \tilde{q} + s$, and $\Lambda \succ 0$, $\tilde{q}(t)$ will be the solution of a stable first-order system. Therefore, $\tilde{q}(t), \dot{\tilde{q}} \to 0$ as $t\in\infty$.

\section{Additional Experimental Results}
\label{sec:additional-fig}
In this section, we present additional results from our numerical evaluations.
The plots below illustrate the trajectory tracking performance of the GD-based adaptation method from~\cite{richards2023control} and our proposed MD-based adaptation for various disturbances from $w = 2 \; [m/s]$ to $w = 10 \; [m/s]$.
We observe that at lower wind speed $w = 2.0, 4.0 \; [m/s]$, both the GD-based adaptation and our proposed MD-based adaptation methods result in good tracking performance.
But under larger wind disturbance $w \ge 6.0 \; [m/s]$, our method (with $p = 2.2$) tracks the reference trajectory much more closely and sees very little degradation of performance from higher wind speed.
Recall that all of our sample trajectories are sampled with wind speed at most $w = 6.0 \; [m/s]$.
Therefore, these experiments demonstrate the superior generalization capabilities of our MD-based adaptation method.

\begin{figure*}[!ht]
    \centering
    \includegraphics[width=0.8\textwidth]{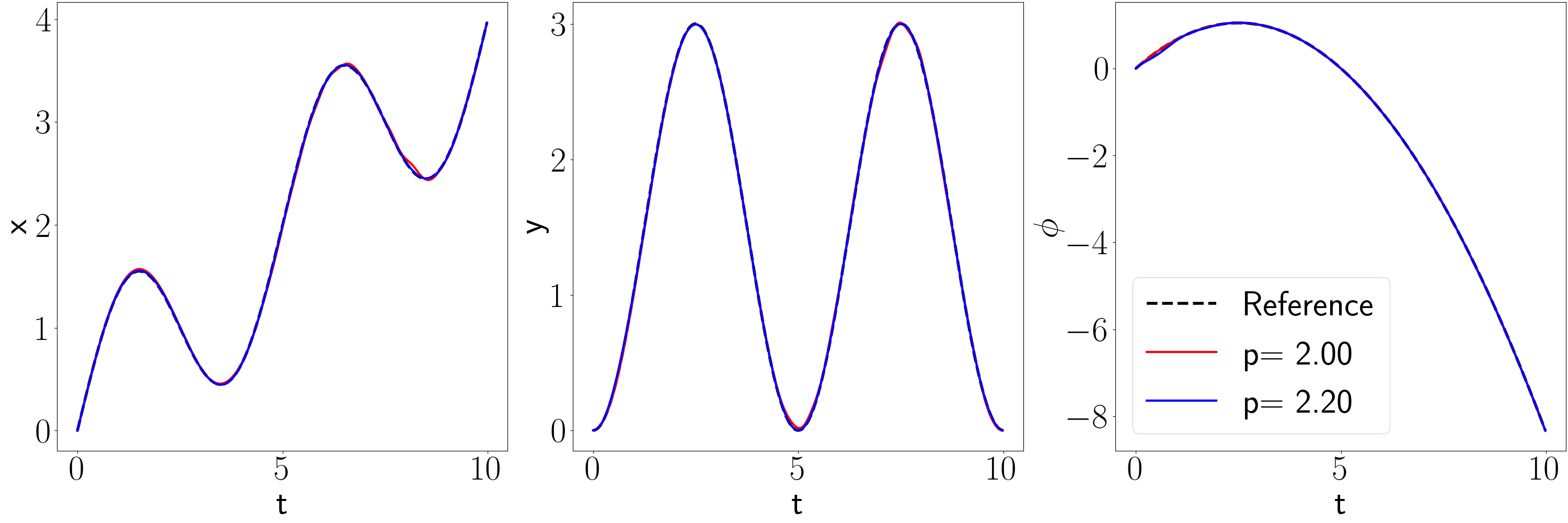}
    \caption{State time-histories under disturbance $w=2.0$ [$m/s$].}
\end{figure*}

\begin{figure*}[!ht]
    \centering
    \includegraphics[width=0.8\textwidth]{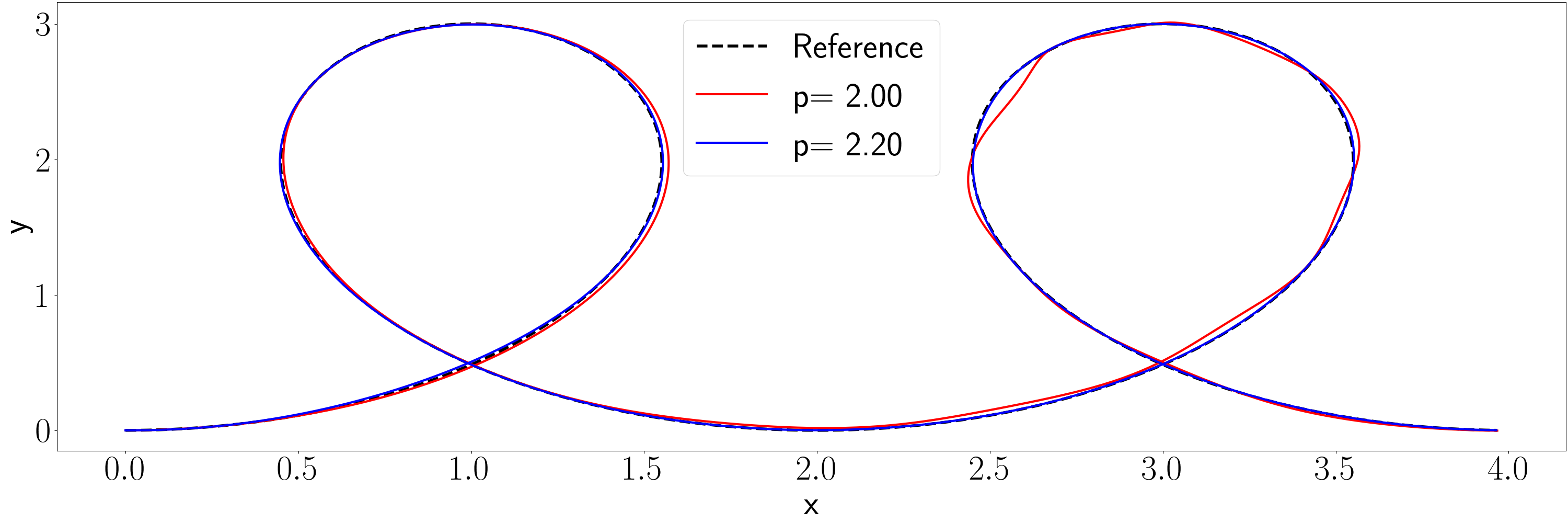}
    \caption{$x-y$ Phase plot in the case of $w=2.0$ [$m/s$].}
\end{figure*}

\begin{figure*}[!ht]
    \centering
    \includegraphics[width=0.8\textwidth]{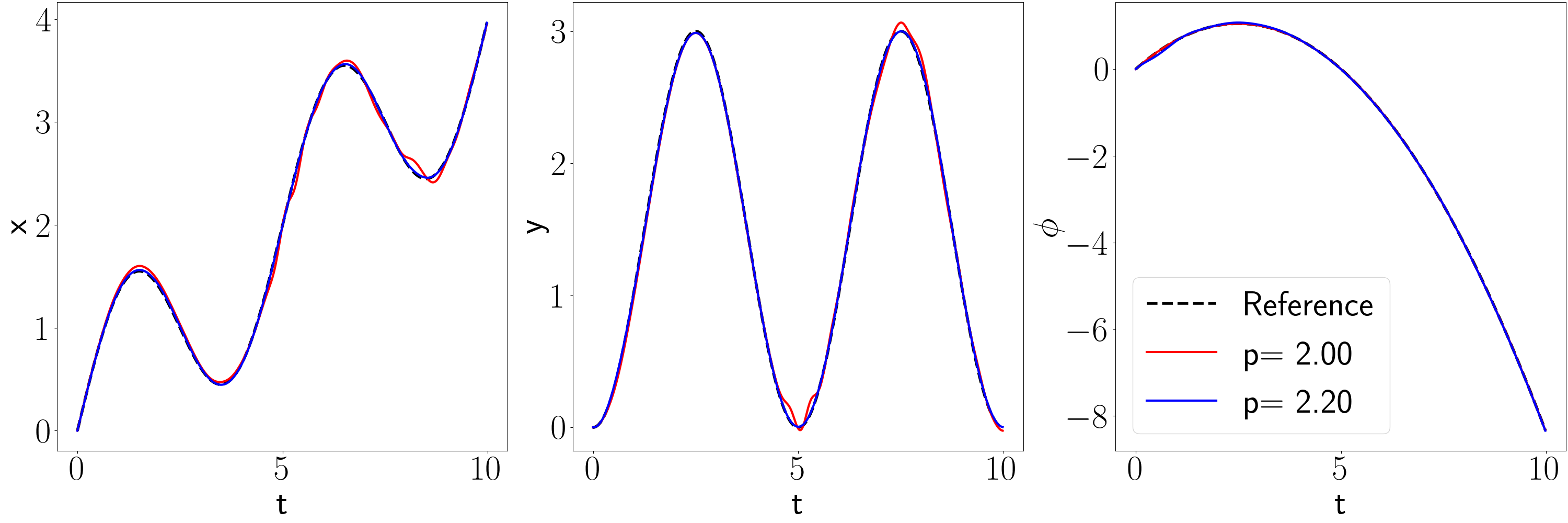}
    \caption{State time-histories under disturbance $w=4.0$ [$m/s$].}
\end{figure*}

\begin{figure*}[!ht]
    \centering
    \includegraphics[width=0.8\textwidth]{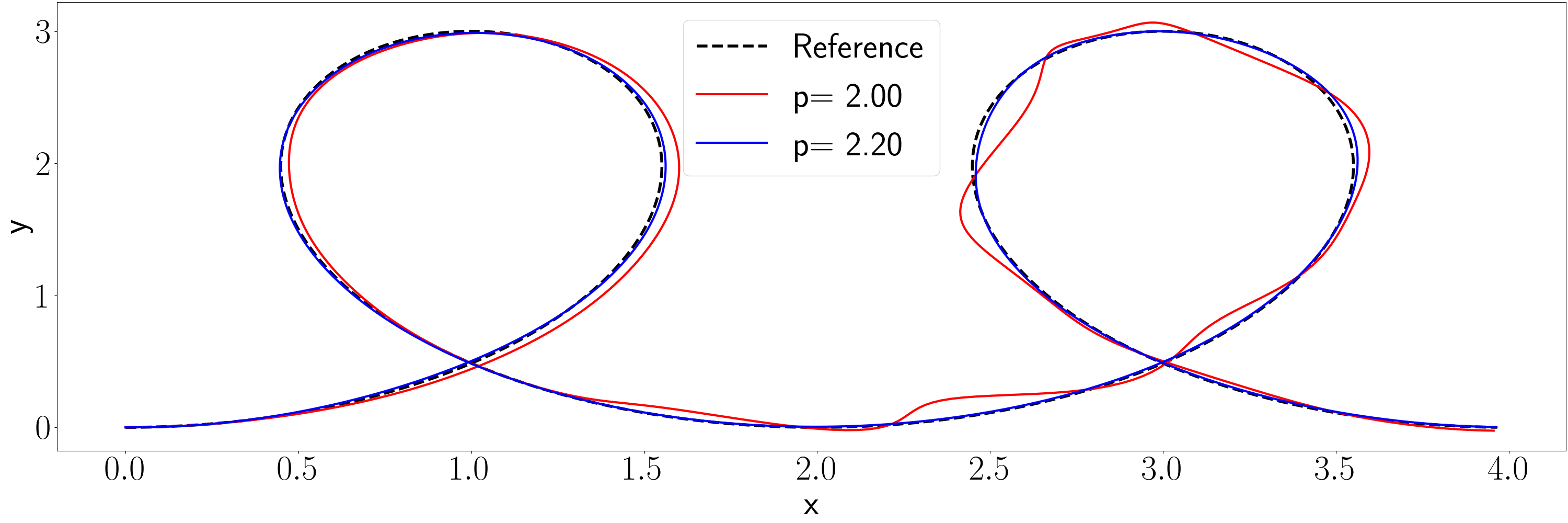}
    \caption{$x-y$ Phase plot in the case of $w=4.0$ [$m/s$].}
\end{figure*}

\begin{figure*}[!ht]
    \centering
    \includegraphics[width=0.8\textwidth]{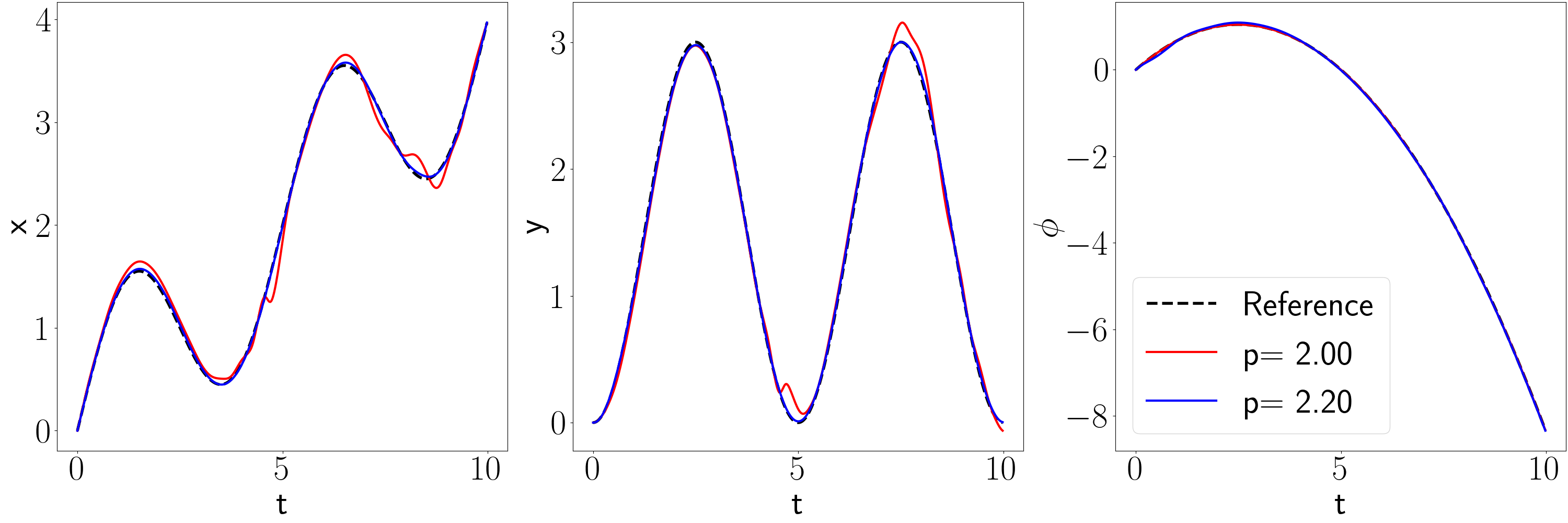}
    \caption{State time-histories under disturbance $w=6.0$ [$m/s$].}
\end{figure*}

\begin{figure*}[!ht]
    \centering
    \includegraphics[width=0.8\textwidth]{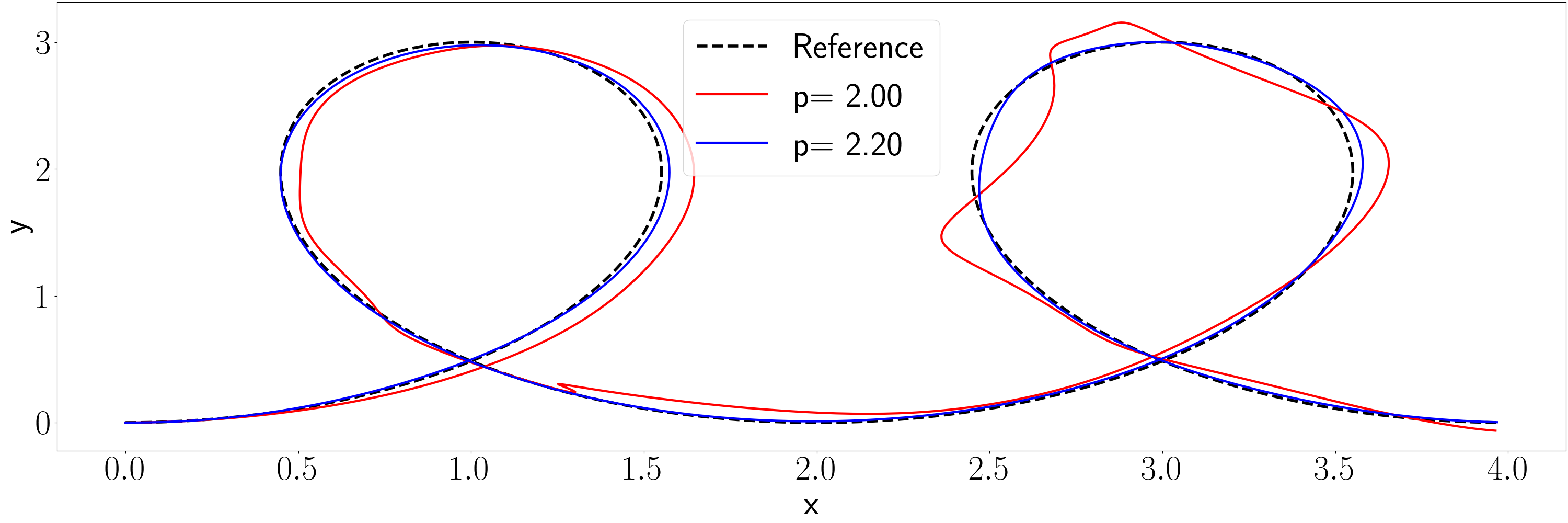}
    \caption{$x-y$ Phase plot in the case of $w=6.0$ [$m/s$].}
\end{figure*}

\begin{figure*}[!ht]
    \centering
    \includegraphics[width=0.8\textwidth]{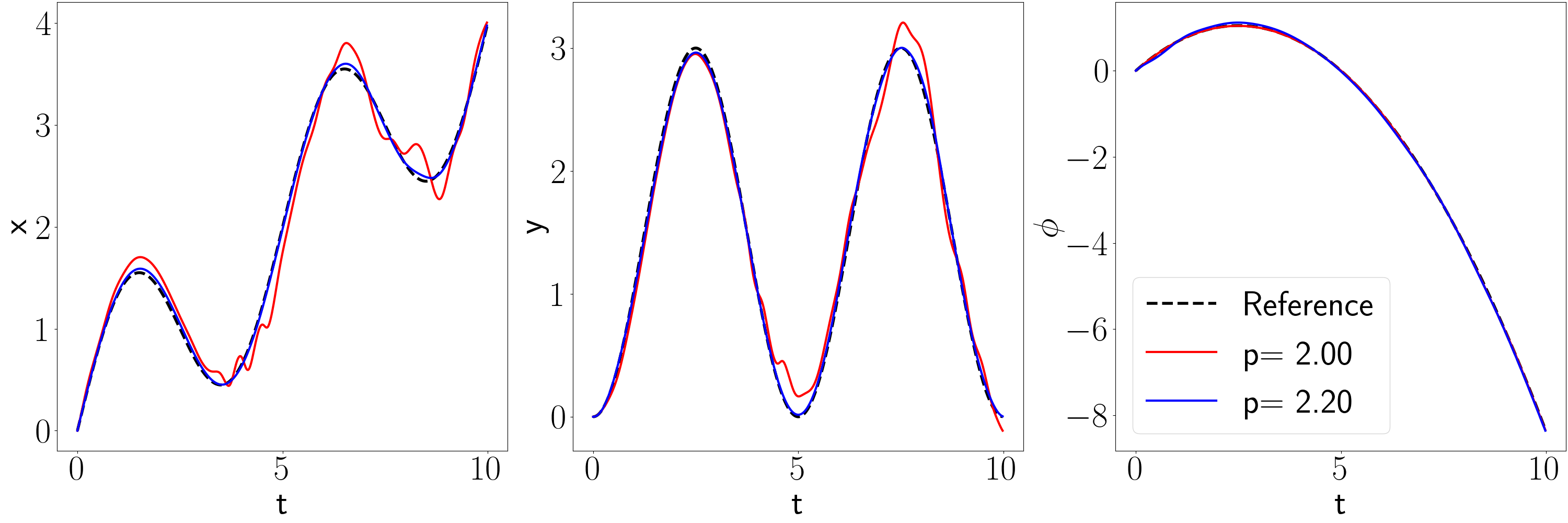}
    \caption{State time-histories under disturbance $w=8.0$ [$m/s$].}
\end{figure*}

\begin{figure*}[!ht]
    \centering
    \includegraphics[width=0.8\textwidth]{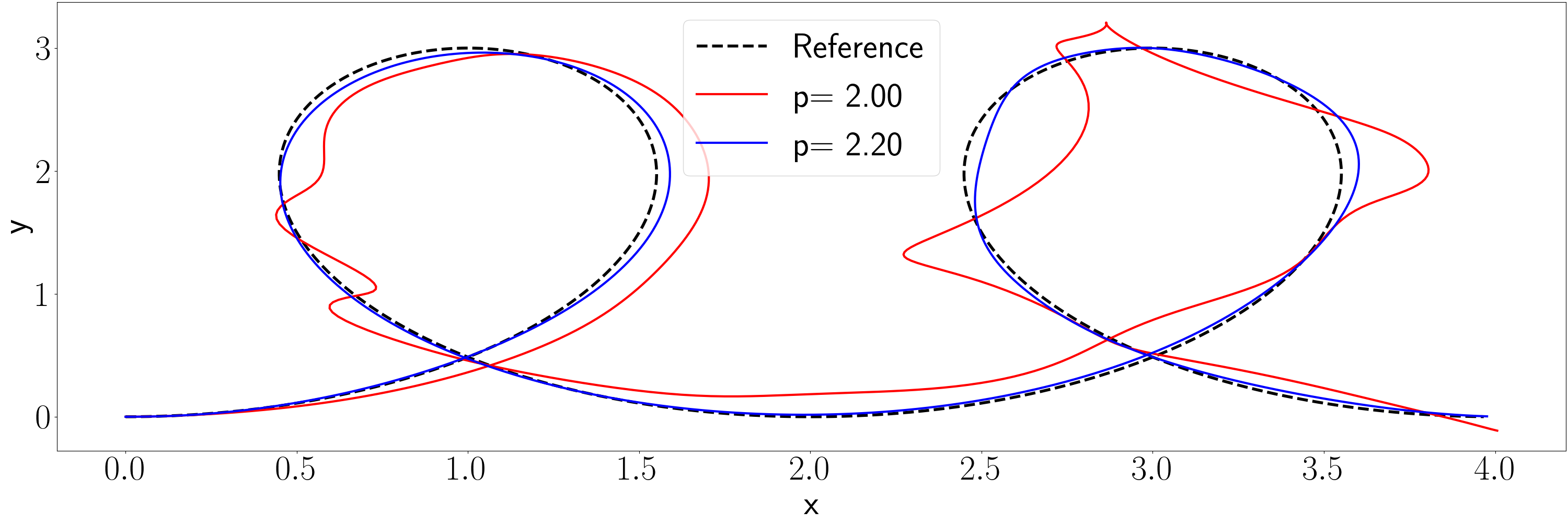}
    \caption{$x-y$ Phase plot in the case of $w=8.0$ [$m/s$].}
\end{figure*}

\begin{figure*}[!ht]
    \centering
    \includegraphics[width=0.8\textwidth]{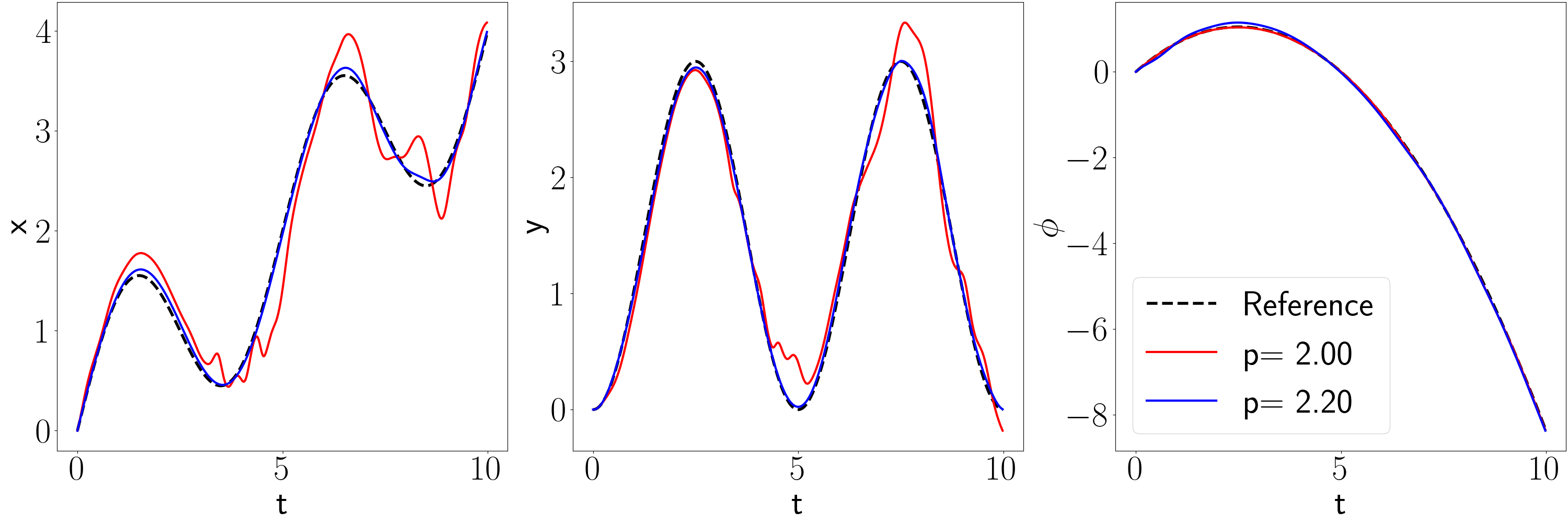}
    \caption{State time-histories under disturbance $w=10.0$ [$m/s$].}
\end{figure*}

\begin{figure*}[!ht]
    \centering
    \includegraphics[width=0.8\textwidth]{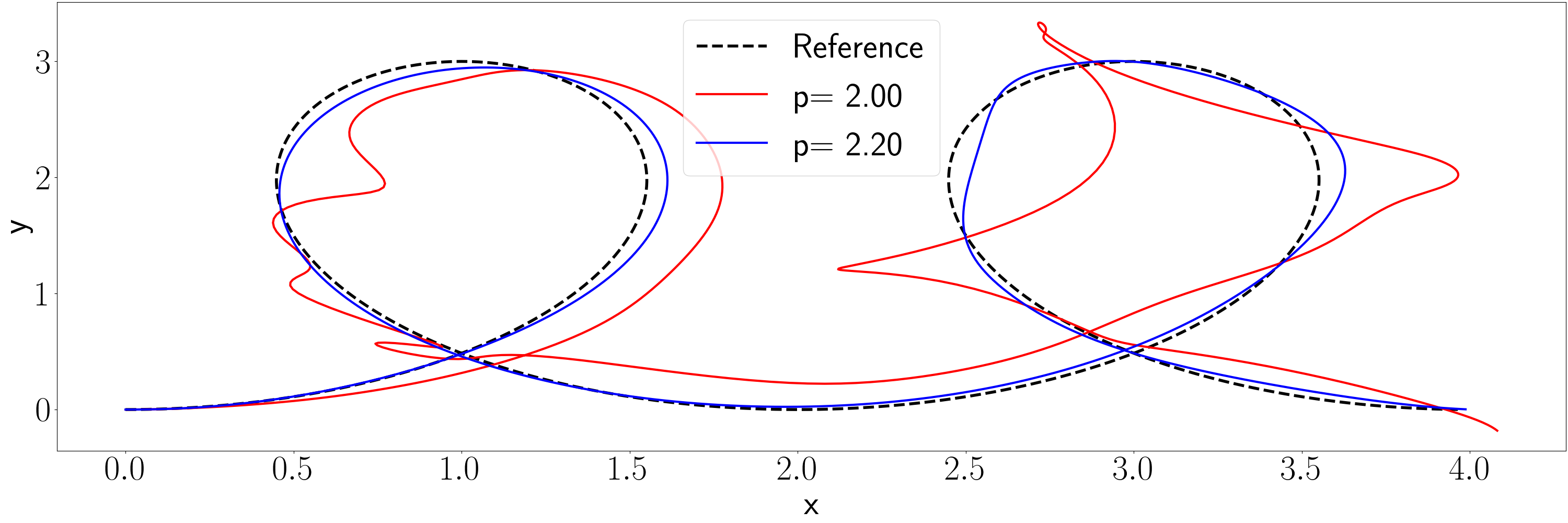}
    \caption{$x-y$ Phase plot in the case of $w=10.0$ [$m/s$].}
\end{figure*}

\end{document}